\documentclass[11pt,leqno,fleqn]{article}

\newcommand{\Comment}[1]{\relax}
\usepackage{caption}
\usepackage{subcaption}
\usepackage{amsthm}
\usepackage{latexsym}
\usepackage{amsmath}
\usepackage{amssymb}
\usepackage{amsfonts}
\usepackage{mathtools}
\usepackage{algorithm}
\usepackage{amsbsy}
\usepackage{diagrams}
\usepackage{enumitem} %  customize \begin{enumerate} ... \end{enumerate}
\usepackage[usenames]{color}
\usepackage[colorlinks=false, %
           citecolor=red,filecolor=blue,linkcolor=blue,urlcolor=blue,
           linktoc=page,
           pagebackref=true
           ]{hyperref} % for hyperlinks
\usepackage{afterpage}  % used in the command \afterpage{\clearpage}.
\usepackage{float} % for option [H] in float placement,
\usepackage{times}
\usepackage{bm}
\usepackage{stmaryrd} % required for \llbracket and \rrbracket
\usepackage{MnSymbol}  % for \llangle and \rrangle
\usepackage{mathrsfs}  % required for \mathscr (``script''), to
\usepackage{graphicx}
\usepackage[all]{xy}
\usepackage{fancybox}
\usepackage{alltt}

\newcommand{\Hide}[1]{}

\newif
\ifnote
\notetrue
\newif
\ifTR
\TRfalse

\newtheorem{theorem}{Theorem}
\newtheorem{lemma}[theorem]{Lemma}

\newtheorem{proposition}[theorem]{Proposition}

\newtheorem{restxxx}[theorem]{Restriction}

\newtheorem{agreexxx}[theorem]{Agreement}

\newtheorem{termxxx}[theorem]{Terminology}

\newtheorem{notxxx}[theorem]{Notation}

\newtheorem{assumxxx}[theorem]{Assumption}

\newtheorem{convenxxx}[theorem]{Convention}

\newtheorem{exaxxx}[theorem]{Example}

\newtheorem{exexxx}[theorem]{Exercise}

\newtheorem{remxxx}[theorem]{Remark}

\newtheorem{openxxx}[theorem]{Open Problem}
\newtheorem{conjxxx}[theorem]{Conjecture}

\newtheorem{defxxx}[theorem]{Definition}
\newenvironment{definition}[1]{\begin{defxxx}[\emph{#1}]\rm}%
{\hfill\QED\end{defxxx}}
\newtheorem{procxxx}[theorem]{Procedure}
{\hfill\QED\end{procxxx}}

\newtheorem{Prxxx}[theorem]{Proof}
{\end{Prxxx}} % {\hfill\QED\end{Prxxx}}

  {\addtolength{\leftskip}{#1}\addtolength{\rightskip}{#2}}{\par}
\newcommand{\Set}[1]{\{ #1 \}}
\newcommand{\SET}[1]{\bigl\{ #1 \bigr\}}
\newcommand{\A}{{\cal A}}
\newcommand{\B}{{\cal B}}

\newcommand{\X}{{\cal X}}

\newcommand{\bigOO}[1]{{\cal O}(#1)} % {\mathscr{O}(#1)}

\newcommand{\Let}[3]%
    {\textbf{\textsf{let}}\ {#1}\,{#2}\ \textbf{\textsf{in}}\;{#3}\,}
\newcommand{\Try}[3]%
    {\textbf{\textsf{try}}\ {#1} {#2}\ \textbf{\textsf{in}}\;{#3}\;}
\newcommand{\Mix}[3]%
    {\textbf{\textsf{mix}}\ {#1} {#2}\ \textbf{\textsf{in}}\;{#3}\;}
\newcommand{\LET}[3]%
    {\textbf{\textsf{let}}^{\bm{*}}\ {#1} {#2}\ \textbf{\textsf{in}}\;{#3}\;}
\newcommand{\Letrec}[3]%
    {\textbf{\textsf{letrec}}\ {#1} {#2}\ \textbf{\textsf{in}}\;{#3}\;}

\newcommand{\bridges}[2]{{\partial}_{#1}(#2)} % {{\text{\em deg}}(#1)} 
 
\newcommand{\degreeSym}{{\text{\em degree}}} 

\newcommand{\degr}[2]{{\degreeSym}_{#1}(#2)}

\newcommand{\heightSym}{{\text{\em height}}} 
\newcommand{\height}[2]{{\heightSym}_{#1}(#2)}

\newcommand{\ie}{\textit{i.e.}}

\newcommand{\QED}{{\Large $\square$}} 
\newcommand{\size}[1]{|\,#1\,|}

\newcommand{\set}[1]{\overline{#1}}
\newcommand{\power}[1]{\mathscr{P}(#1)}
\newcommand{\minbisect}{\text{\rm MinBisection}} % {\text{\rm MinCut}}
\newcommand{\cliquecov}{\text{\rm CliqueCover}} % {\text{\rm MinCut}}
\newcommand{\spacing}[2]{
  \renewcommand{\baselinestretch}{#2}
  \small\normalsize #1
  \setlength{\parskip}{0.1\baselineskip}
  \settowidth{\parindent}{xxxx}
  \setlength{\parindent}{#2\parindent}
  \setlength{\leftmargini}{\parindent}
  \setlength{\leftmarginii}{\parindent}
  \setlength{\leftmarginiii}{\parindent}
  \setlength{\footnotesep}{#2\footnotesep}
}

\begin{document}

\spacing{\normalsize}{0.98}
\setcounter{page}{1}     % \thispagestyle{empty} has to be inserted right
                         % after \maketitle to suppress the page numbering
                         % on the title page
\setcounter{tocdepth}{1} % to suppress subsections and subsubsections
                         % in the table of contents.
\ifTR
  \pagenumbering{roman} % for title page and table-of-contents page
\else
\fi

\title{Efficient Reassembling of Graphs, Part 2: The Balanced Case}

\author{Saber Mirzaei %
           \thanks{Partially supported by NSF awards CCF-0820138
           and CNS-1135722.} \\
        Boston University \\
        \ifTR Boston, Massachusetts \\
        \href{mailto:kfoury@bu.edu}{kfoury{@}bu.edu}
        \else \fi
\and
       Assaf Kfoury %
          \footnotemark[1]\\
       Boston University  \\
        \ifTR Boston, Massachusetts \\
        \href{mailto:smirzaei@bu.edu}{smirzaei{@}bu.edu}
        \else \fi
}

\ifTR
   \date{\today}
\else
   \date{} %
\fi
\maketitle
  \ifTR
     \thispagestyle{empty} % it has to be inserted right after \maketitle
                           % in order to suppress the page numbering
  \else
  \fi

\vspace{-.3in}
  \begin{abstract}
   %% abstract.tex

\noindent
The \emph{reassembling of a simple connected graph} $G = (V,E)$ with
$n = \size{V} \geqslant 1$ vertices is an
abstraction of a problem arising in earlier studies of network
analysis. The reassembling process has a simple formulation
(there are several equivalent formulations) relative to a binary tree $\B$ 
-- its so-called \emph{reassembling tree}, with root node at the top 
and $n$ leaf nodes at the bottom -- where every cross-section corresponds 
to a partition of $V$ (a block in the partition is a node in the 
cross-section) such that:
\begin{itemize}[itemsep=0pt,parsep=2pt,topsep=5pt,partopsep=0pt] 
\item the bottom (or first) cross-section (\ie, all the leaves) is the 
      finest partition of $V$ with $n$ one-vertex blocks,
\item the top (or last) cross-section (\ie, the root) is the coarsest partition
      with a single block, the entire set $V$,
\item a node (or block) in an intermediate cross-section (or partition)
      is the result of merging its two children nodes (or blocks) in the
      cross-section (or partition) below it.
\end{itemize}
The \emph{edge-boundary degree} of a block $A$ of vertices
is the number of edges with one endpoint in $A$ and one endpoint
in $(V-A)$. The \textbf{maximum} edge-boundary degree encountered during 
the reassembling process is what we call the 
$\bm{\alpha}$\textbf{-measure} of the reassembling, and
the \textbf{sum} of all edge-boundary degrees is its
$\bm{\beta}$\textbf{-measure}. The $\alpha$-optimization
(resp. $\beta$-optimization) of the reassembling of $G$ is to determine a
reassembling tree $\B$ that minimizes its $\alpha$-measure 
(resp. $\beta$-measure).

\medskip % \smallskip % 
\noindent
There are different forms of reassembling, depending on the shape
of the reassembling tree $\B$. In an earlier report, we studied
\emph{linear reassembling}, which is the case when the height of
$\B$ is $(n-1)$. In this report, we study \emph{balanced reassembling},
when $\B$ has height $\lceil \log n\rceil$. In a forthcoming report,
we study \emph{general reassembling}, which is the case when the height
of $\B$ can be any number between $(n-1)$ and $\lceil \log n\rceil$. 

\medskip % \smallskip % 
\noindent
The two main results in this report are the NP-hardness of
$\alpha$-optimization and $\beta$-optimization of balanced
reassembling. The first result is obtained by a sequence of
polynomial-time reductions from \emph{minimum bisection of graphs}
(known to be NP-hard), and the second by a sequence of polynomial-time
reductions from \emph{clique cover of graphs} (known to be NP-hard).

%% END of abstract.tex

  \end{abstract}

\ifTR
    \newpage
    \tableofcontents
    \newpage
    \pagenumbering{arabic}
\else
    \vspace{-.2in}
\fi

\section{Introduction}
\label{sect:intro}
   %% introduction.tex

A more extensive introduction to the problem of graph reassembling, including
the motivation for studying it, is in our earlier report~\cite{kfouryReassembling}. We here limit
ourselves to a brief review of the \emph{balanced} case of the problem.

\paragraph{Problem Statement.}
Let $G = (V,E)$ be a simple (no self-loops, no multi-edges),
connected, undirected graph, with $\size{V} = n\geqslant 1$ vertices
and $\size{E} = m$ edges.  The \emph{reassembling} of
$G$ can be defined relative to a binary tree $\B$
-- one root node at the top, $n$ leaf nodes at the bottom --
where every cross-section corresponds to a partition of $V$
(a block in the partition is a node in the cross-section) such that:
\begin{itemize}[itemsep=0pt,parsep=2pt,topsep=5pt,partopsep=0pt]
\item the bottom (or first) cross-section (\ie, all the leaves) is the
      finest partition of $V$ with $n$ one-vertex blocks,
\item the top (or last) cross-section (\ie, the root) is the coarsest partition
      with a single block, the entire set $V$,
\item a node (or block) in an intermediate cross-section (or partition)
      is the result of merging its two children nodes (or blocks) in the
      cross-section (or partition) below it.
\end{itemize}
For convenience, we say \emph{vertices} to refer to the vertices of $G$
and \emph{nodes} to refer to those of the tree $\B$. We call $G$ the
\emph{input graph} and $\B$ the \emph{reassembling tree}.

The height of $\B$ is at least $\lceil\log n\rceil$ and at most $(n-1)$.
We say the reassembling of $G$ according to $\B$ is \emph{balanced} if
$\B$'s height is $\lceil\log n\rceil$; this is the case when, at every
level of the reassembling, there is a maximum number of pairs of blocks
which are each merged into a block at the next level up.

If $A$ and $B$ are disjoint nonempty sets of vertices in $G$,
a \emph{bridge} between $A$ and $B$ is an edge whose two endpoints are
one in $A$ and one in $B$.  We write $\bridges{}{A,B}$ to denote the
set of all bridges between $A$ and $B$. In case $B = V-A$, the set
$\bridges{}{A,B}$ is the same as the edge cut-set determined by the
cut $(A,V-A)$, and instead of writing $\bridges{}{A,V-A}$ we write
$\bridges{}{A}$ for simplicity. The \emph{edge-boundary degree} of $A$
is the number of bridges with only one endpoint in $A$, \ie,
$\size{\bridges{}{A}}$.

Various optimization problems can be associated with graph
reassembling. Two such optimizations are the following, identified
by the letters $\alpha$ and $\beta$ throughout. We want to
determine a reassembling tree $\B$ for which:
\begin{itemize}[itemsep=0pt,parsep=2pt,topsep=5pt,partopsep=0pt]
\item[($\alpha$)]
    the \textbf{maximum} edge-boundary degree encountered during
    reassembling is minimized, or
\item[($\beta$)]
    the \textbf{sum} of all edge-boundary degrees encountered during
    reassembling is minimized.
\end{itemize}
Initially, before we start reassembling, we always set the $\alpha$-measure
$M_{\alpha}$ to the \textbf{maximum} of all the vertex degrees, \ie,
$\max\,\Set{\degr{}{v}\,|\,v\in V}$,
and we set the $\beta$-measure $M_{\beta}$
to the \textbf{sum} of the vertex degrees, \ie,
$\sum\,\Set{\degr{}{v}\,|\,v\in V}$.
During reassembling, after we merge disjoint nonempty subsets $A$
and $B$, we update the $\alpha$-measure $M_{\alpha}$ to:
$\max\,\SET{M_{\alpha},\size{\bridges{}{A\cup B}}}$,
and the $\beta$-measure $M_{\beta}$ to:
$\bigl(M_{\beta} + \size{\bridges{}{A\cup B}}\bigr)$. The reassembling process
terminates when we reach the root of the reassembling tree $\B$.

\paragraph{About Terminology.}
Our choice of names comes from applications we studied in earlier
report~\cite{kfouryReassembling}, where we had to disassemble and reassemble flow
networks in such a way that the resulting $\alpha$-measure and
$\beta$-measure were minimized. However, \emph{graph reassembling} as
here defined can also be viewed as belonging to the large family
of \emph{graph embedding} problems, whereby a graph $G$ (often called
a \emph{source graph} in the graph-theoretic literature) is embedded
into another graph $G'$ (often called a \emph{host graph}) in such a
way that various optimization measures (motivated by applications) are
minimized or maximized.

For our reassembling problem, a source graph is what we call
an \emph{input graph} $G = (V,E)$, and a host graph is a rooted binary
tree, which we call a \emph{reassembling tree} $\B$. Moreover, the
embedding in our problem takes a special form (namely, as explained
above, there is a one-one correspondence between the vertices of $G$
and the leaves of $\B$, the full set $V$ is mapped to the root of
$\B$ and etc.).

We stick to our terminology, in part to be consistent with our earlier
reports, where we study or use graph reassembling. But we also avoid
in this report the use of concepts and terminology that are incidental
to our graph-theoretic examination.

\newcommand{\CliqueCover}[1]{#1\text{-CliqeCover}}

\paragraph{Main Results.}
We restrict attention to the balanced case of graph reassembling in
this report.  We prove that $\alpha$-optimization and
$\beta$-optimization of balanced reassembling are both NP-hard
problems. We obtain these results by showing that:
\begin{itemize}[itemsep=0pt,parsep=2pt,topsep=5pt,partopsep=0pt]
\item
   there is a sequence of several polynomial-time reductions from
   \emph{minimum bisection of graphs} ($\minbisect$) to
   $\alpha$-optimization of balanced reassembling of graphs,
\item
   there is a sequence of several polynomial-time reductions
   from \emph{clique cover of graphs} ($\cliquecov$)
   to $\beta$-optimization of balanced
   reassembling of graphs.
\end{itemize}
Both $\minbisect$ and $\cliquecov$ have been extensively studied: They
are both NP-hard in general~\cite{garey1976some, karp1972reducibility}.
This leaves open the problem of identifying
classes of graphs, whether of practical or theoretical significance,
for which there are low-degree polynomial-time solutions for
our two optimization problems.

\paragraph{Organization of the Report.}

In Section~\ref{sect:problem} we give precise formal definitions of
several notions underlying our examination of the balanced case.
Section~\ref{sect:alpha-optimization} presents the NP-hardness of
$\alpha$-optimization, and Section~\ref{sect:beta-optimization}
the NP-hardness of $\beta$-optimization, for the balanced
case of graph reassembling. We believe some of the intermediate
reductions in Section~\ref{sect:alpha-optimization} and
Section~\ref{sect:beta-optimization} are of independent
interest. Section~\ref{sect:future} is a short wrap-up of related and
future work.

\section{Notational Conventions and Preliminary Definitions}
\label{sect:problem}
  %% problem-statement.tex

We agree on notations and conventions for this report, and reproduce
enough from the earlier~\cite{kfouryReassembling} to make the present one self-contained.

All graphs are simple (no self-loops, no multi-edges)
and undirected. With no loss of generality, we assume that all graphs
have each an even number of vertices.

We denote the set of vertices of a graph $G$ by
$V(G)$, and its set of edges by $E(G)$. If there is an edge
connecting two vertices $v,w\in V(G)$, we write $\set{v\,w}$ for
the two-element set representing that edge. Unless explicitly
stated otherwise, we reserve the letter ``$n$'' for the number
of vertices and the letter ``$m$'' for the number of edges.

For convenience, we may refer to the graph $G$ by writing $G=(V,E)$
instead of $G=\bigl(V(G),E(G)\bigr)$, and also refer to the sizes of
$V(G)$ and $E(G)$ by writing $\size{V}$ and $\size{E}$ instead of
$\size{V(G)}$ and $\size{E(G)}$.

A graph $H$ is a \emph{subgraph} of the graph $G$ iff $V(H)\subseteq
V(G)$ and $E(H)\subseteq E(G)$. In this report we only need to
consider subgraphs that are each induced by a subset of vertices.  The
subgraph $H$ of $G$ is said to be \emph{induced} by a subset $A\subseteq V(G)$
iff $V(H) = A$ and $E(H) = \Set{\,\set{v\,w}\in E(G)\,|\,v,w\in A\,}$.
We write $G[A]$ to denote the subgraph of $G$ induced by the subset
$A$ of vertices

We write ${\bridges{G}{A}}$ or ${\bridges{}{G[A]}}$ or, if $G$ is clear
from the context, just ${\bridges{}{A}}$ to denote
the \emph{edge-boundary} of the subgraph $G[A]$, \ie,
${\bridges{}{A}} =
\Set{\,\set{v\,w}\in E(G)\,|\,v\in A\text{ and } w\not\in A\,}$.

More generally, if $A$ and $B$ are disjoint subsets of $V(G)$,
the set of \emph{bridges} between $A$ and $B$
is denoted by ${\bridges{G}{A,B}}$ or, if $G$ is clear from the context,
just ${\bridges{}{A,B}}$ which is the set
$\Set{\,\set{v\,w}\in E(G)\,|\,v\in A\text{ and } w\in B\,}$.
Thus, ${\bridges{}{A}}$ is the same as ${\bridges{}{A,V-A}}$.

\newcommand{\graphs}[1]{{\textsc{graphs}}(#1)}

We write $\graphs{2^{\bm{*}}}$ to refer to the class of all simple
graphs which have each a power-of-$2$ number of vertices.

\begin{definition}{Minimum Bisection Problem}
\label{def:min-bisection}
A \emph{bisection} of the graph $G = (V,E)$ is a partition
$\Set{A,B}$ of $V(G)$ with two blocks of equal size:
\[
   \size{A} = \size{B},\quad A\cap B = \varnothing, \quad
   \text{and}\ \ A\cup B = V .
\]
(Our running assumption is that $V$ has an even number of vertices and can
therefore be partitioned into two equal-size blocks.)
We say the bisection $\Set{A,B}$
is of \emph{type} $(X,Y)$ iff $X,Y \subseteq V(G)$ such that:
\begin{itemize}[itemsep=0pt,parsep=0pt,topsep=2pt,partopsep=0pt]
\item
  either $ X\cap Y = \varnothing, \quad
   X \subseteq A, \quad \text{and}\ \ Y \subseteq B $,
\item
  or \quad\;  $ X\cap Y = \varnothing, \quad
   X \subseteq B, \quad \text{and}\ \ Y \subseteq A $ .
\end{itemize}
A \emph{minimum bisection} $\Set{A,B}$ of $G$ is one that minimizes
the set of bridges ${\bridges{}{A,B}}$ between $A$ and $B$, \ie,
\[
   \size{\bridges{}{A,B}}\ \leqslant
   \ \min\,\SET{\,\size{\bridges{}{A',B'}}\;\bigl|
              \;\text{$\Set{A',B'}$ is a bisection of $G$}\,}.
\]
The \emph{minimum bisection problem}, $\minbisect$, asks for
finding a minimum bisection of a given input graph $G$.
$\minbisect$ is known to be NP-hard~\cite{garey1976some}.

We write $\minbisect(2^{\bm{*}})$ for the restriction of $\minbisect$ to the
class $\graphs{2^{\bm{*}}}$.  Lemma~\ref{lem:NP-hardness-of-bisection}
asserts that $\minbisect(2^{\bm{*}})$ is also NP-hard.
\end{definition}

\newcommand{\colorability}[1]{#1\text{-Colorability}}

\begin{definition}{Clique Cover Problem}
\label{def:clique-cover}
Given an integer $k\geqslant 1$, a $k$-\emph{clique cover} of the graph
$G = (V,E)$ is a partition of $V(G)$ into $k$ disjoint
subsets $\Set{V_1,\ldots,V_k}$ such that each of the induced
subgraphs $G[V_1],\ldots,G[V_k]$ is a complete graph  (or a clique in $G$).

The $k$-\emph{clique cover problem}, $\CliqueCover{k}$ in shorthand,
is a (yes,no) question that asks whether a graph $G$ has a $k$-clique cover.
$\CliqueCover{k}$ is polynomial-time solvable for $k = 1,2$, and
is known to be NP-complete for every $k\geqslant 3$~\cite{karp1972reducibility}.
\end{definition}

Our notion of a \emph{reassembling} of graph $G =(V,E)$ presupposes
the notion of a \emph{binary tree} over the finite set
$V(G) = \Set{v_1,\ldots,v_n}$.  Our definition of \emph{binary trees} is
not standard, but is more convenient for our purposes.%
   \footnote{A standard definition of a binary tree $T$ makes
   $T$ a subset of the set of finite binary strings $\Set{0,1}^*$ such that:
\begin{itemize}[itemsep=0pt,parsep=0pt,topsep=2pt,partopsep=0pt]
\item $T$ is prefix-closed, \ie, if $t\in T$ and $u$ is a prefix of $t$,
      then $u\in T$.
\item Every node has two children, \ie, $t\,0\in T$ iff $t\,1\in T$
      for every $t\in \Set{0,1}^*$.
\end{itemize}
The \emph{root node} of $T$ is the empty string $\varepsilon$,
and a \emph{leaf node} of $T$ is a string $t\in T$ without children, \ie,
both $t\,0\not\in T$ and $t\,1\not\in T$.
}

\begin{definition}{Binary trees}
\label{defn:binaryTrees}
An \emph{(unordered) binary tree} $\B$ over $V(G)=\Set{v_1,\ldots,v_n}$
is a collection of non-empty subsets of $V(G)$ satisfying three
conditions:
\begin{enumerate}[itemsep=0pt,parsep=2pt,topsep=5pt,partopsep=0pt]
\item For every $v\in V(G)$, the singleton set $\Set{v}$ is in $\B$.
\item The full set $V(G)$ is in $\B$.
\item For every $X\in\B$, there is a unique $Y\in\B$ such that:
      $X\cap Y=\varnothing$ and $(X\cup Y)\in\B$.
\end{enumerate}
The \emph{leaf nodes} of $\B$ are the singleton sets $\Set{v}$, and
the \emph{root node} of $\B$ is the full set $V(G)$. Depending on the
context, we may refer to the members of $\B$ as its \emph{nodes} or as
its \emph{clusters}. Several expected properties of $\B$, reproducing
familiar ones of a standard definition, are stated in the next two
propositions.
\end{definition}

Proofs for Propositions~\ref{prop:propertiesOne} and~\ref{prop:propertiesTwo}
are straightforward and here omitted. Details can be found in the
earlier report~\cite{kfouryReassembling}.

\begin{proposition}[Properties of binary trees]
\label{prop:propertiesOne}
Let $\B$ be a binary tree as in Definition~\ref{defn:binaryTrees},
let $v\in V(G)$, and let:
\begin{equation*}
\label{eq:path}
\tag{$\dag$}
    \Set{v} = X_0\ \subsetneq\ X_1 \ \subsetneq X_2
                  \ \subsetneq\ \cdots\ \subsetneq\ X_p = V(G)
\end{equation*}
be a maximal sequence of nested clusters from $\B$. We then have:
\begin{enumerate}[itemsep=0pt,parsep=2pt,topsep=5pt,partopsep=0pt]
\item The sequence in~\eqref{eq:path} is uniquely defined, \ie,
      every maximal nested sequence starting from $\Set{v}$ is the same.
\item For every cluster $Y\in\B$, if $\Set{v}\subseteq Y$, then
      $Y\in\Set{X_0,\ldots,X_p}$.
\item There are pairwise disjoint clusters
      $\Set{Y_0,\ldots,Y_{p-1}} \subseteq \B$ such that, for every
      $0 \leqslant i < p$:
\[    X_i \cap Y_{i} = \varnothing \quad\text{and}\quad
      X_i \cup Y_{i} = X_{i+1}.
\]
\end{enumerate}
\end{proposition}

\noindent
Based on this proposition, we use the following terminology:
\begin{itemize}[itemsep=0pt,parsep=2pt,topsep=5pt,partopsep=0pt]
\item We call the sequence in~\eqref{eq:path}, which is unique by part 1,
      the \emph{path} from leaf node $\Set{v}$ to root node $V(G)$.
\item Every cluster containing $v$ is one of the nodes along this unique
      path, according to part 2.
\item In part 3, $Y_i$ is the \emph{sibling node} of
      $X_i$, and both are the \emph{children nodes} of the
      \emph{parent node} $X_{i+1}$.
\end{itemize}

\begin{proposition}[Properties of binary trees]
\label{prop:propertiesTwo}
Let $\B$ be a binary tree as in Definition~\ref{defn:binaryTrees}.
We then have:
\begin{enumerate}[itemsep=0pt,parsep=2pt,topsep=5pt,partopsep=0pt]
\item For all clusters $X,Y\in\B$,
      if $X\cap Y\neq\varnothing$ then $X\subseteq Y$
      or $Y\subseteq X$.
\item For every cluster $X\in\B$, the sub-collection
      of clusters ${\B}_{X} := \Set{\,Y\in\B\;|\;Y\subseteq X\,}$ is a
      binary tree over $X$, with root node $X$.
\item $\B$ is a collection of $(2n-1)$ nodes/clusters.
\end{enumerate}
\end{proposition}

Let $\theta$ be a map from $V(G)$ to $V(G)$. We extend $\theta$
to a map on every subset $X\subseteq V(G)$ by defining:
$\theta(X) := \Set{\,\theta(v)\;|\;v\in X}$, and on
every set of subsets $\X \subseteq\power{V(G)}$ by defining
$\theta(\X) := \Set{\,\theta(X)\;|\; X\in\X}$.
If $\theta$ is a bijection on $V(G)$, then
$\theta$ is also a bijection on $\power{V(G)}$.

\begin{proposition}[Properties of binary trees]
\label{prop:propertiesThree}
Let $\B$ be a binary tree as in Definition~\ref{defn:binaryTrees}.
If $\theta: V(G)\to V(G)$ is a bijection on $V(G)$, then:
\begin{enumerate}[itemsep=0pt,parsep=2pt,topsep=5pt,partopsep=0pt]
   \item $\theta(\B)$ is a binary tree, \ie, the collection
         of subsets in $\theta(\B)$ satisfies the three conditions
         in Definition~\ref{defn:binaryTrees}.
   \item $\B$ and $\theta(\B)$ are isomorphic trees, \ie,
     \begin{itemize}[itemsep=0pt,parsep=2pt,topsep=2pt,partopsep=0pt]
       \item $\theta$ maps the leaf nodes and the root node of $\B$ to the
             leaf nodes and the root node of $\theta(\B)$.
       \item $\theta$ maps a path
    $\Set{v} = X_0 \subsetneq X_1 \subsetneq \cdots\subsetneq X_p = V(G)$
    in $\B$ to a path \\
    $\Set{\theta(v)} = \theta(X_0) \subsetneq
     \theta(X_1) \subsetneq \cdots
     \subsetneq\theta(X_p) = V(G)$ in $\theta(\B)$.
       \item $\theta$ maps a pair of sibling nodes $X$ and $Y$ in $\B$
           to a pair of sibling nodes $\theta(X)$ and $\theta(Y)$
           in $\theta(\B)$.
     \end{itemize}
\end{enumerate}
In words, every bijection on $V(G)$ lifts to an isomorphism between
binary trees on $V(G)$.
\end{proposition}

\begin{proof}
Straightforward from Definition~\ref{defn:binaryTrees}.
\end{proof}

A common measure for a standard definition of binary trees is $\heightSym$,
which is represented in the following, corresponding to our notion of binary trees in
Definition~\ref{defn:binaryTrees}:
\begin{alignat*}{5}
    &\height{}{\B}\ &&:=
    \ && \max\,\bigl\{\,p\;\bigl|\;&& \text{there is $v\in V(G)$ such that }
\\
    & && && &&
      \Set{v} = X_0 \subsetneq X_1 \subsetneq\cdots\subsetneq X_p = V(G)
      \text{ is a maximal sequence of nested clusters}\,\bigr\} .
\end{alignat*}
For a particular node/cluster $X \in\B$, the subtree of $\B$ rooted at
$X$ is ${\B}_X$, by part 2 in Proposition~\ref{prop:propertiesTwo}.
The height of $X$ in $\B$ is therefore
$\height{\B}{X} := \height{}{{\B}_X}$.

The binary tree $\B$ over $\Set{v_1,\ldots,v_n}$
is \emph{balanced} if $\height{}{\B} = \lceil\log n\rceil$. If $n$
is a power of $2$, say $n = 2^p$, then  $\height{}{\B} = p$.

\begin{definition}{Graph Reassembling}
\label{defn:binary-reassembling}
\label{defn:graph-reassembling}
A \emph{reassembling of the graph $G = (V,E)$} is simply
defined by a pair $(G,\B)$ where $\B$ is a binary tree over $V(G)$, as in
Definition~\ref{defn:binaryTrees}.

Given a reassembling $(G,\B)$ of $G$, two measures are of particular
interest for our analysis, namely, for every node/cluster $X\in\B$, the
\emph{edge-boundary degree} (or simply the \emph{degree})
of $X$ and the \emph{height} of $X$ in $\B$:
\[
  \degr{G,\B}{X}\ :=\ \size{\bridges{G}{X}}
\quad\text{and}\quad
  \height{G,\B}{X}\ :=\ \height{\B}{X} .
\]
If the context makes clear the reassembling $(G,\B)$ -- respectively,
the binary tree $\B$ -- relative
to which of these measures are defined, we write $\degr{}{X}$ and
$\height{}{X}$ -- or $\degr{G}{X}$ and $\height{G}{X}$, respectively --
instead of $\degr{G,\B}{X}$ and $\height{G,\B}{X}$.%
     \footnote{ Our reassembling of $G$ can be viewed
        as ``hierarchical clustering'' of $G$, similar to a method of
        analysis in data mining, though used for a different purpose. Our
        reassembling mimicks so-called ``agglomerative, or
        bottom-up, hierarchical clustering'' in data mining.}

We say the reassembling $(G,\B)$ is \emph{balanced} if the underlying
binary tree $\B$ is balanced.
\end{definition}

\begin{definition}{Isomorphic Graph Reassemblings}
\label{defn:isomorphic-graph-reassemblings}
Let $(G,\B)$ and $(G,{\B}')$ be two reassemblings of the
same graph $G$. We say $(G,\B)$ and $(G,{\B}')$
are \emph{isomorphic reassemblings} if there is a
bijection $\theta : V(G)\to V(G)$
% $\theta : \Set{v_1,\ldots,v_n} \to \Set{v_1,\ldots,v_n}$
satisfying two conditions:
\begin{enumerate}[itemsep=0pt,parsep=2pt,topsep=5pt,partopsep=0pt]
   \item  ${\B}' = \theta(\B)$.
   \item  For every node $X\in\B$, we have
          $\degr{G,\B}{X} = \degr{G,{\B}'}{\theta(X)}$.
\end{enumerate}
In words, for $(G,\B)$ and $(G,{\B}')$ to be isomorphic,
not only do we require that the reassembling trees $\B$ and $\B'$
be isomorphic (Proposition~\ref{prop:propertiesThree}), but
also that the degrees of every node $X$ and its image $\theta(X)$ be equal.
\end{definition}

The following definition repeats a definition
in Section~\ref{sect:intro} more formally.

\begin{definition}{Measures on the reassembling of a graph}
\label{def:different-measures}
Let $(G,\B)$ be a reassembling of $G$.
We define the measures $\alpha$ and $\beta$
on $(G,\B)$ as follows:
\begin{alignat*}{10}
    & \alpha(G,\B) &&:=
     \ &&\max\, &&\SET{\,\degr{G,\B}{X}\; \bigl| \; X\in\B \,},
\\[1.2ex]
    & \beta(G,\B) &&:=
     \ &&\sum\, &&\SET{\,\degr{G,\B}{X}\; \bigl| \; X\in\B \,}.
\end{alignat*}
An optimization problem arises with the minimization of each of these
measures. We say the reassembling $(G,\B)$ is \emph{$\alpha$-optimal} iff:
\[
  \alpha (G,\B)\ =\ \min\,
  \SET{\,\alpha (G,{\B}')\;\bigl|
  \;{\B}'\text{ is a binary tree over $V(G)$}\,} .
\]
We say $(G,\B)$ is an \emph{$\alpha$-optimal balanced reassembling} iff:
\[
  \alpha (G,\B)\ = \ \min\,
  \SET{\,\alpha (G,{\B}')\;\bigl|
  \;{\B}'\text{ is a balanced binary tree over $V(G)$}\,}.
\]
\textbf{Important note:} When we say ``$(G,\B)$ is an $\alpha$-optimal
balanced reassembling,'' we mean $(G,\B)$ is $\alpha$-optimal
among all \emph{balanced} reassemblings only, we
do not mean that $(G,\B)$ is $\alpha$-optimal among all reassemblings
and that $(G,\B)$ happens to be balanced.

We leave to the reader the obvious similar
definition of what it means for $(G,\B)$ to
be a \emph{$\beta$-optimal} balanced reassembling.
\end{definition}

% \section{Examples}
% \label{sect:examples}
  % \input{examples}

\section{$\alpha$-Optimization of Balanced Reassembling Is NP-Hard}
\label{sect:alpha-optimization}
    %% alpha-optimization

Consider an arbitrary graph $G = (V,E)$.  Let $p$ be the smallest
integer such that $2^{p}\geqslant\size{V} = n \geqslant 2$. It follows
that $n + r = 2^{p}$ for some even integer $r$ such that $0\leqslant r
< n$. Another way of identifying $r$ is to say it is the
smallest even integer such $n+r$ is a power of $2$.
Our standing assumption is that $n$ is even.

\newcommand{\GG}{\mathbb{G}}

Let $q = (n/2) +r$. Following standard notation,
$K_{q}$ denotes the complete graph over
$q$ vertices. For Lemma~\ref{lem:preliminary}, we
construct an augmented graph $\GG$ consisting of the original
$G$ together with two disjoint copies of $K_{q}$, which we denote by
the separate letters $H$ and $I$ for clarity. More precisely,
\begin{alignat*}{8}
  & V(\GG)\ &&=\ && V(G)\ && \uplus\ && V(H)\ && \uplus\ && V(I),
\\[.8ex]
  & E(\GG)\ &&=\ && E(G) && \uplus && E(H) && \uplus && E(I)\ \uplus
   \  \Set{\,\set{v\,w}\:|\; v\in V(G)\text{ and } w\in V(H)\cup V(I)\,},
\end{alignat*}
where we write ``$\uplus$'' for ``disjoint union''. We thus assemble
the new $\GG$ by connecting every vertex in $G$ with all the vertices
in $H$ and $I$. There are no edges between $H$ and $I$.
There is a total of $n+q+q = 2n + 2r = 2^{p+1}$ vertices in $\GG$.
Thus, the graph $\GG$ has between $2n$ and $3n$ vertices and
is a member of the class $\graphs{2^{\bm{*}}}$.

\begin{lemma}
\label{lem:preliminary}
Consider the graph $\GG$ as defined in the preceding paragraph.
Every minimum bisection of $\GG$ must be of type
$\bigl(V(H),V(I)\bigr)$ -- see Definition~\ref{def:min-bisection} --
schematically shown in Figure~\ref{fig:preliminary}.
\end{lemma}

\begin{proof}
We consider a bisection $\Set{A,B}$ of
$\GG$ which is \emph{not} of type $\bigl(V(H),V(I)\bigr)$, and then
show that it cannot be a minimum bisection.

Because $\Set{A,B}$
is a bisection, we have $\size{A} = \size{B} = n+r$. And
because $\size{V(H)} = \size{V(I)} = (n/2)+r$, it is never the case
that $A \cap (V (H) \uplus V (I)) = \varnothing$ or
$B \cap (V (H) \uplus V (I)) = \varnothing$. Hence,
there are two possible cases, one shown in
Figure~\ref{fig:preliminary1} and one in
Figure~\ref{fig:preliminary2}. In Figure~\ref{fig:preliminary1}, the
vertices of both $H$ and $I$ appear on both sides of the bisection.
In Figure~\ref{fig:preliminary2}, the vertices of $H$ are all on the
same side of the bisection, while the vertices of $I$ appear on both
sides of the bisection.  In these figures, $\Set{X,X',Y,Y',Z,Z'}$ is a
$6$-block partition of $V(\GG)$, defined by:
\begin{alignat*}{10}
  & X\ && =\ && A\cap V(G)\quad && \text{and} \quad && X'\ && =\ && B\cap V(G)
   \qquad && \text{(the vertices of subgraph $G$)},
\\
  & Y\ && =\ && A\cap V(H)\quad && \text{and} \quad && Y'\ && =\ && B\cap V(H)
   \qquad && \text{(the vertices of subgraph $H$)},
\\
  & Z\ && =\ && A\cap V(I)\quad && \text{and} \quad && Z'\ && =\ && B\cap V(I)
   \qquad && \text{(the vertices of subgraph $I$)}.
\end{alignat*}
In each of the two cases shown in  Figures~\ref{fig:preliminary1}
and~\ref{fig:preliminary2}, we need to prove that
$\size{\bridges{}{A,B}}$ can be decreased by moving an appropriate
number of vertices of the subgraphs $H$ and $I$ from one side of the
bisection to the other side.

\paragraph{Case 1.} This is the case in Figure~\ref{fig:preliminary1}.
With no loss of generality, suppose:
\[
  \size{X} \ \geqslant \ \size{X'} .
\]
Because $\size{A} = \size{B}$, this forces the inequality:
\[
  \size{Y} + \size{Z} \ \leqslant
  \ \size{Y'} + \size{Z'} .
\]
Because $\size{Y}+\size{Y'} =
\size{Z}+\size{Z'}$, this in turn forces one or both
of the following inequalities:
\begin{equation*}\tag{$\dag$}   % \label{eq:preliminary}
  \size{Y} \leqslant \size{Z'}
  \quad\text{or}\quad
  \size{Z} \leqslant \size{Y'}  .
\end{equation*}
With no loss of generality, assume the first inequality
$\size{Y} \leqslant \size{Z'}$ in ($\dag$) holds.

\Hide{
If both inequalities in ($\dag$)
hold, we can assume with no loss of generality the following inequality holds:
\begin{equation*}\tag{$\ddag$}   % \label{eq:preliminary}
  \size{Z'} - \size{Y}\ \leqslant\ \size{Y'} - \size{Z} .
  % \size{Z'-Y} \leqslant \size{Y'-Z}
\end{equation*}
If only one inequality in ($\dag$) holds, say,
$\size{Y} \leqslant \size{Z'}$, then again the inequality in ($\ddag$)
must hold, otherwise we would have
$\size{Z'} - \size{Y} > \size{Y'} - \size{Z}$ or equivalently
$\size{Z} + \size{Z'} > \size{Y} + \size{Y'}$, which is a
contradiction.
}

Let $\size{Y} = k \geqslant 1$.  Select an
arbitrary subset $U \subseteq Z'$ such that $\size{U} = k$.
We define a new bisection $\Set{\tilde{A},\tilde{B}}$ of $\GG$
by moving: (1) all the vertices of $Y$ from the $A$-side to the $B$-side,
and (2) all the vertices of $U$ from the $B$-side to the $A$-side.
Specifically, let:
\[
  \tilde{A}\ =\ (A-Y)\cup U\quad\text{and}\quad
  \tilde{B}\ =\ (B-U)\cup Y .
\]
The resulting set of edges connecting $\tilde{A}$ and $\tilde{B}$ is:
\begin{alignat*}{6}
 & \bridges{}{\tilde{A},\tilde{B}}\ =
 \ && \Bigl(\,\bridges{}{A,B}\ -\ \Bigl( \bridges{}{Y,X'}
                     \,\cup\, \bridges{}{Y,Y'}
                     \,\cup\, \bridges{}{U,X}
                     \,\cup\, \bridges{}{U,Z}
                     \Bigr) \Bigr)
\\
 & && \bigcup\ \ \Bigl(\, \bridges{}{Y,X}
              \,\cup\,\bridges{}{U,X'}
              \,\cup\,\bridges{}{U,Z'-U}\, \Bigr) .
\end{alignat*}
Because
$\size{\bridges{}{U,X}} = \size{\bridges{}{Y,X}}$
and $\size{\bridges{}{U,X'}} = \size{\bridges{}{Y,X'}}$,
it follows:
\[
  \size{\bridges{}{\tilde{A},\tilde{B}}}\ =
  \ \size{\bridges{}{A,B}}
  \ - \ \size{\bridges{}{Y,Y'}}
  \ - \ \size{\bridges{}{U,Z}}
  \ + \ \size{\bridges{}{U,Z'-U}}
\]
With the fact that $\size{Y} = \size{U}$, the following
inequality must hold:
\[
   \size{\bridges{}{Y,Y'}}\ \geqslant\ \size{\bridges{}{U,Z'-U}}
\]
otherwise, if it did not, we would have that $\size{Y'} < \size{Z'-U}
= \size{Z'}-\size{U} = \size{Z'}-\size{Y}$, in turn implying that
$\size{Y}+\size{Y'} < \size{Z'}$, which is a contradiction. Hence,
\[
  \size{\bridges{}{\tilde{A},\tilde{B}}}\ \geqslant
  \ \size{\bridges{}{A,B}} \ - \ \size{\bridges{}{U,Z}}
  \ >\ \size{\bridges{}{A,B}}.
\]
We conclude that
$\size{\bridges{}{\tilde{A},\tilde{B}}} < \size{\bridges{}{A,B}}$.

\paragraph{Case 2.}
This is the case in Figure~\ref{fig:preliminary2}.
It cannot be that $\size{X} < \size{X'}$, because if it
were so, it would imply $\size{X} < (n/2)$ which, with the
fact that $\size{Z\cup Z'} = (n/2) + r$, would in turn imply that
$\size{A} < n+r$, thus contradicting the hypothesis that
$\Set{A,B}$ is a bisection of $\GG$. Hence, it must be that:
\[
  \size{X} \ \geqslant \ \size{X'} \ \geqslant \ n/2 .
\]
The largest possible size of $X$, which is $n-1$, corresponds to the
smallest possible size of $Z$ which, because $\size{X}+\size{Z} = n+r$,
must therefore be $r+1$. Corresponding to the smallest possible
size of $Z$ is the largest possible size of $Z'$, which is therefore
$\size{V(I)} - (r+1) = (n/2) - 1$. Hence, it is always the case that
$\size{Z'} < \size{X}$.

We now proceed in a way similar to \textbf{Case 1}.  Let
$\size{Z'} = k \geqslant 1$. Select an arbitrary subset $U\subseteq X$
such that $\size{U} = k$ and $\bridges{}{U,X'}\neq\varnothing$.
The new bisection $\Set{\tilde{A},\tilde{B}}$ of $\GG$ is obtained
by moving: (1) all the vertices of $Z'$ from the $B$-side to the $A$-side,
and (2) all the vertices of $U$ from the $A$-side to the $B$-side.
Specifically, let:
\[
  \tilde{A}\ =\ (A-U)\cup Z'\quad\text{and}\quad
  \tilde{B}\ =\ (B-Z')\cup U .
\]
The resulting set of edges connecting $\tilde{A}$ and $\tilde{B}$ is:
\begin{alignat*}{6}
 & \bridges{}{\tilde{A},\tilde{B}}\ =
 \ \ && \Bigl(\,\bridges{}{A,B}\ -\ \Bigl( \bridges{}{Z',Z}
                     \,\cup\, \bridges{}{Z',X-U}
                     \,\cup\, \bridges{}{Y',U}
                     \,\cup\, \bridges{}{X',U}
                     \Bigr) \Bigr)
\\
 & && \bigcup\ \ \Bigl(\, \bridges{}{Z',X'}
              \,\cup\,\bridges{}{Z',U}
              \,\cup\,\bridges{}{Z,U}
              \,\cup\,\bridges{}{X-U,U}\, \Bigr) .
\end{alignat*}
Because $\size{\bridges{}{Y',U}} = \size{\bridges{}{Z\cup Z',U}}
          = \size{\bridges{}{Z,U}} + \size{\bridges{}{Z',U}}$, it follows that:
\begin{alignat*}{6}
 & \size{\bridges{}{\tilde{A},\tilde{B}}}\ =
 \ \ && \size{\bridges{}{A,B}}\ -\ \size{\bridges{}{Z',Z}}
                     \,-\, \size{\bridges{}{Z',X-U}}
                     \,-\, \size{\bridges{}{X',U}}
\\
 & && +\, \size{\bridges{}{Z',X'}}\, +\, \size{\bridges{}{X-U,U}} .
\end{alignat*}
Because $\size{\bridges{}{Z',X-U}} \geqslant \size{\bridges{}{X-U,U}}$, we
obtain the inequality:
\begin{alignat*}{6}
 & \size{\bridges{}{\tilde{A},\tilde{B}}}\ \leqslant
 \ \ && \size{\bridges{}{A,B}}\ -\ \size{\bridges{}{Z',Z}}
                     \,-\, \size{\bridges{}{X',U}}
                     \,+\, \size{\bridges{}{Z',X'}} .
\end{alignat*}
Because $\size{\bridges{}{Z',Z}} \geqslant \size{\bridges{}{Z',X'}}$,
it follows that:
\begin{alignat*}{6}
 & \size{\bridges{}{\tilde{A},\tilde{B}}}\ \leqslant
 \ \ && \size{\bridges{}{A,B}}\
                     \,-\, \size{\bridges{}{X',U}} .
\end{alignat*}
Because $\size{\bridges{}{X',U}}\neq 0$, we conclude
that $\size{\bridges{}{\tilde{A},\tilde{B}}} < \size{\bridges{}{A,B}}$.
\end{proof}
\begin{figure}
    \centering
        \begin{subfigure}[b]{0.25\textwidth}
                \includegraphics[scale=.5]
                {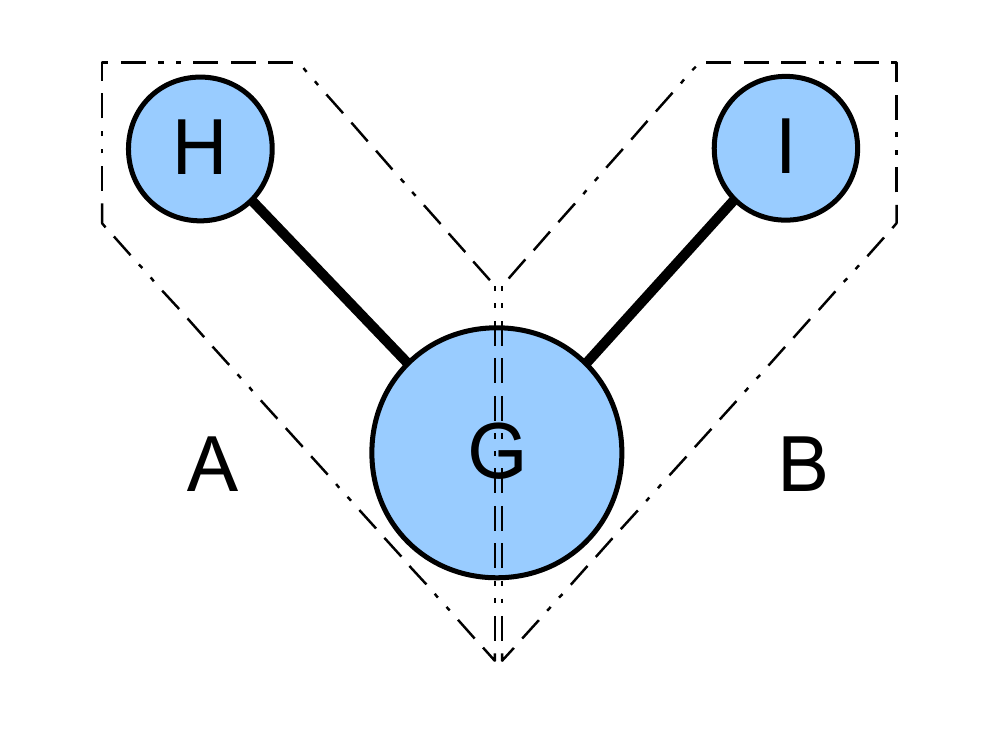}
                \caption{Graph $G$ augmented with two complete graph $H$ and $I$.}
                \label{fig:preliminary}
        \end{subfigure}
        \qquad
        \begin{subfigure}[b]{0.25\textwidth}
                \includegraphics[scale=.55]
                {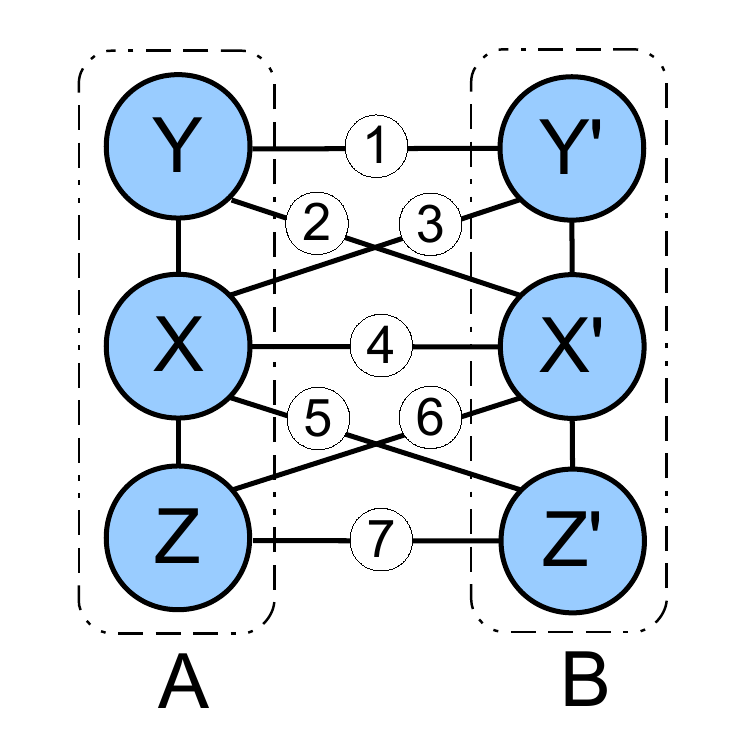}
                \caption{An arbitrary bisection for the augmented graph $\GG$.}
                \label{fig:preliminary1}
        \end{subfigure}
                \qquad
        \begin{subfigure}[b]{0.25\textwidth}
                \includegraphics[scale=.55]
                {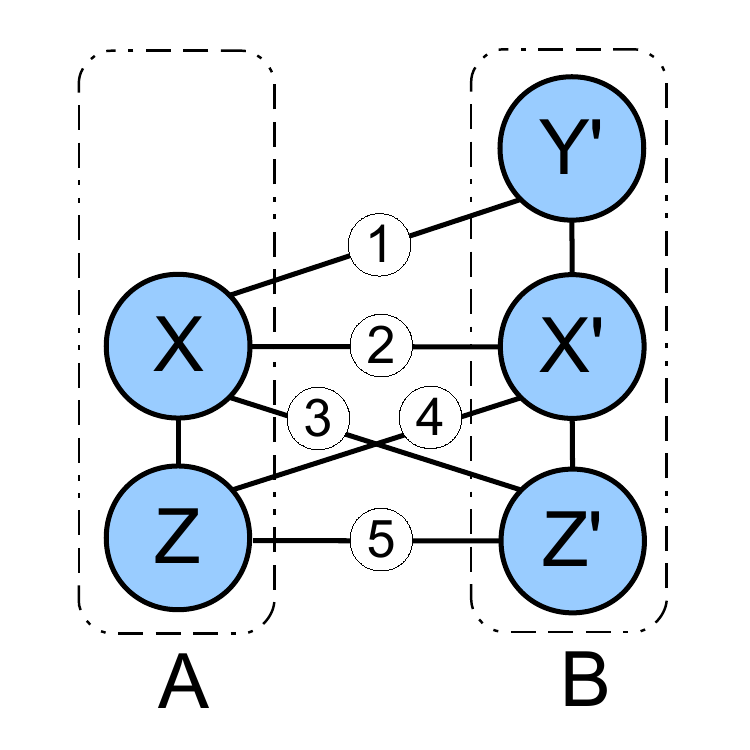}
                \caption{An improved bisection for the augmented graph $\GG$ based on the arbitrary bisection~\ref{fig:preliminary1}.}
                \label{fig:preliminary2}
        \end{subfigure}
        \caption{For the statement of Lemma~\ref{lem:preliminary}. }
        \label{fig:preliminaries}
\end{figure}
\iffalse
\begin{figure}[H] % [hc!]

    \noindent
    \begin{minipage}[b]{0.45\linewidth}
      \includegraphics[scale=.5]%
       {Graphics/preliminary-1.pdf}
    \end{minipage}

\caption{Graph $G$ augmented with two complete graph $H$ and $I$. For the statement of Lemma~\ref{lem:preliminary}. }
\label{fig:preliminary}
\end{figure}

\begin{figure}[H] % [hc!]

    \noindent
    \begin{minipage}[b]{0.45\linewidth}
      \includegraphics[scale=.5]%
       {Graphics/preliminary-2.pdf}
    \end{minipage}

\caption{An arbitrary bisection for the augmented graph $\GG$. For the proof of Lemma~\ref{lem:preliminary}. }
\label{fig:preliminary1}
\end{figure}

\begin{figure}[H] % [hc!]

    \noindent
    \begin{minipage}[b]{0.45\linewidth}
      \includegraphics[scale=.5]%
       {Graphics/preliminary-3.pdf}
    \end{minipage}

\caption{An improved bisection for the augmented graph $\GG$ based on the arbitrary bisection~\ref{fig:preliminary1}. For the proof of Lemma~\ref{lem:preliminary}. }
\label{fig:preliminary2}
\end{figure}
\fi
%
%
\begin{lemma}
\label{lem:NP-hardness-of-bisection}
$\minbisect(2^{\bm{*}})$ is an NP-hard problem.
\end{lemma}

\begin{proof}
We use the same notation as in the proof of
Lemma~\ref{lem:preliminary}.  Given that subgraphs $H$ and $I$ of
$\GG$ have an equal number $q = (n/2)+r$ of vertices,
Lemma~\ref{lem:preliminary} implies we can reduce, in polynomial time,
$\minbisect$ for an arbitrary graph $G$ to $\minbisect(2^{\bm{*}})$ for
graph $\GG$. Hence, a minimum bisection for $\GG$ induces a minimum
bisection for the given $G$.  Hence, the NP-hardness of $\minbisect$
in general implies the NP-hardness of $\minbisect(2^{\bm{*}})$.
\end{proof}

\begin{lemma}
\label{lem:alpha-optimal-balanced}
Let $G = (V,E)$ be a graph in the class $\graphs{2^{\bm{*}}}$ and let
$\GG$ be the augmented graph of $G$ as in Lemma~\ref{lem:preliminary}.
Let $(\GG,\B)$ be a balanced reassembling of $\GG$ where $\B$
is a binary tree over $V(\GG)$ (see Definition~\ref{defn:binaryTrees}).
Let $A$ and $B$ be the two children nodes of the root node $V(\GG)$ in $\B$.
% and ${\B}_A$ and ${\B}_B$ the subtrees of $\B$ rooted at $A$ and $B$.

\medskip
\noindent
\textbf{Conclusion:}
If $(\GG,\B)$ is an $\alpha$-optimal balanced reassembling, then
$\Set{A,B}$ is a minimum bisection of $\GG$.
\end{lemma}

\begin{proof}
Because $G\in\graphs{2^{\bm{*}}}$, each of the complete graphs
$H$ and $I$ has $(n/2)$ vertices, \ie, in contrast to the proof
of Lemma~\ref{lem:preliminary}, here $r = 0$.
The augmented graph $\GG$ is also in the class  $\graphs{2^{\bm{*}}}$,
with $\size{V(G)} = n$ and $\size{V(\GG)} = 2n$.
In the given balanced reassembling $(\GG,\B)$, we have
that $\Set{A,B}$ is a bisection of $\GG$, with $\size{A} = \size{B} = n$.
The subtrees ${\B}_A$ and ${\B}_B$ rooted at $A$ and $B$ are each over
$n$ vertices. By Definitions~\ref{defn:graph-reassembling}
and~\ref{def:different-measures}, we have
$\alpha(\GG,\B) \geqslant \size{\bridges{}{A,B}}$.
The conclusion of the lemma will follow from the
fact that $\alpha(\GG,\B) = \size{\bridges{}{A,B}}$, which we
show next.

In the rest of the proof we use the notation and definitions in
the proof of Lemma~\ref{lem:preliminary}. There are two cases,
depending on whether $\Set{A,B}$ is, or is not,
of type $\bigl(V(H),V(I)\bigr)$.

\paragraph{Case 1.}
If $\Set{A,B}$ is of type $\bigl(V(H),V(I)\bigr)$, we can assume
that $V(H) \subseteq A$ and $V(I) \subseteq V(I)$. We
pose:
\[
   X = A\cap V(G), \quad X' = B\cap V(G),
   \quad Y = A\cap V(H), \quad \text{and}\ \ Z = B\cap V(I) .
\]
By construction, $\size{X} = \size{X'} = (n/2)$ and
$\size{Y} = \size{Z} = (n/2)$. Let $C = \size{\bridges{}{X,X'}}$, which
is the value of the bisection of the given graph $G$.
Because $H$ and $I$ are copies of the complete graph $K_{(n/2)}$,
the set of edges connecting $Y$ to $X'$, and
the set of edges connecting $Z$ to $X$, satisfy the equalities:
\[
    \size{\bridges{}{Y,X'}} = \size{\bridges{}{Z,X}} = (n/2)^2 = n^2/4 .
\]
Hence, using the notation of Definitions~\ref{defn:graph-reassembling}
and~\ref{def:different-measures}, we have:
\[
   \size{\bridges{}{A,B}} =
   \degr{\GG,\B}{A} = \degr{\GG,\B}{B} = (n^2/4) + (n^2/4) + C = (n^2/2) + C .
\]
Hence, $\alpha(\GG,\B) \geqslant (n^2/2) + C$. The equality in fact holds
because, as argued next, $\degr{\GG,\B}{S} \leqslant (n^2/2) + C$
for every node/cluster of vertices $S$ in the subtrees ${\B}_A$
and ${\B}_B$.

Let $S$ be a node in ${\B}_A$. (The same argument applies if $S$ is a node
in ${\B}_B$.) If $S = A$, we already know that
$\degr{\GG,\B}{S} \leqslant (n^2/2) + C$. Suppose $S$ is not the root of
${\B}_A$, \ie, $S\neq A$. This implies $\size{S} \leqslant \size{A}/2$.
Let $\size{S} = k$, so that $k\leqslant (n/2)$.
Let $S_1 = S\cap X$ and $S_2 = S\cap Y$, with $\size{S_1} = k_1$
and $\size{S_2} = k_2$, so that also $k_1+k_2 = k \leqslant (n/2)$.
Note that $\bridges{}{S_1,X'} \subseteq \bridges{}{X,X'}$,
so that if $C_1 = \size{\bridges{}{S_1,X'}}$, then $C_1 \leqslant C$.
Also, $\bridges{}{S_2,X'} = \bridges{}{S_2,Z} = \varnothing$.
We conclude:
\[
   \degr{\GG,\B}{S}\ =
   \ \underbrace{k_1 (n/2) + C_1}_{\size{\bridges{}{S_1,B}}}
    + \underbrace{k_1k_2}_{\size{\bridges{}{S_1,S_2}}}
    + \underbrace{k_2 (n/2)}_{\size{\bridges{}{S_2,X'}}}
   \ =\ (k_1+k_2)(n/2) + k_1k_2 + C_1
   \ \leqslant\ (n^2/2) + C .
\]
\paragraph{Case 2.}
Suppose $\Set{A,B}$ is not of type $\bigl(V(H),V(I)\bigr)$.
By Lemma~\ref{lem:preliminary}, $\Set{A,B}$ is not a minimum
bisection and therefore $\size{\bridges{}{A,B}} > (n^2/2) + C$
where $C$ is defined as in Case 1 above. Hence
$\alpha(\GG,\B)  > (n^2/2) + C$. By Case 1,
$(\GG,\B)$ is not an $\alpha$-optimal balanced reassembling.
\end{proof}

\begin{theorem}
\label{thm:alpha-optimization}
For the class of simple undirected graphs $G$, the computation of
$\alpha$-optimal balanced reassemblings $(G,\B)$ is an NP-hard problem.
\end{theorem}

\begin{proof}
If a deterministic polynomial-time algorithm $\A$ existed for computing
an $\alpha$-optimal balanced reassembling for an arbitrary
graph $G$, then $\A$ could be used for computing an $\alpha$-optimal
balanced reassembling for the augmented graph $\GG$ of
an arbitrary graph $G\in\graphs{2^{\bm{*}}}$. Hence,
by Lemma~\ref{lem:alpha-optimal-balanced},
$\A$ could also be used for computing a minimum bisection of $\GG$
in deterministic polynomial time. This would contradict the NP-hardness of
$\minbisect(2^{\bm{*}})$, as asserted by Lemma~\ref{lem:NP-hardness-of-bisection}.
The desired conclusion follows.
\end{proof}

\section{$\beta$-Optimization of Balanced Reassembling Is NP-Hard}
\label{sect:beta-optimization}
    %% beta-optimization.tex

We need to introduce two variations of the $\CliqueCover{k}$ problem
(Definition~\ref{def:clique-cover}).

\begin{definition}{Fixed-Size $\CliqueCover{4}$, Equal-Size $\CliqueCover{4}$}
\label{def:fixed-size-clique-cover}
In the \emph{Fixed-Size $\CliqueCover{4}$ problem} we consider a graph
$G$ together with four positive integers $\Set{n_1,n_2,n_3,n_4}$ such
that $n_1+n_2+n_3+n_4 = \size{V(G)}$ and we ask: Can we partition
$V(G)$ into four disjoint subsets $A_1$, $A_2$, $A_3$ and $A_4$ of
respective sizes $n_1, n_2, n_3$ and $n_4$ such that each of the
induced subgraphs $G[A_1]$, $G[A_2]$, $G[A_3]$ and $G[A_4]$ is a
complete graph (\ie, $A_1$, $A_2$, $A_3$ and $A_4$ are cliques)?

In the \emph{Equal-Size $\CliqueCover{4}$ problem} we consider a graph
$G$ and ask: Can we partition
$V(G)$ into four disjoint subsets $A_1$, $A_2$, $A_3$ and $A_4$ of
equal size, \ie, $\size{A_1}=\size{A_2}=\size{A_3}=\size{A_4}=\size{V(G)}/4$,
such that each of $G[A_1]$, $G[A_2]$, $G[A_3]$ and $G[A_4]$ is a
complete graph (\ie, $A_1$, $A_2$, $A_3$ and $A_4$ are cliques)?
\end{definition}

\begin{lemma}
\label{lem:fixed-size-4-clique-cover}
Fixed-Size $\CliqueCover{4}$ is NP-complete.
\end{lemma}

\begin{proof}
Given a $4$-part partition $\Set{A_1,A_2,A_3,A_4}$ of $V(G)$, 
we can verify in polynomial time that the induced graphs
$\Set{G[A_1],G[A_2],G[A_3],G[A_4]}$ are each complete and that their
sizes are the given $\Set{n_1,n_2,n_3,n_4}$. So the problem is in NP.

We next show that
NP-completeness follows by reduction from $\CliqueCover{4}$,
\ie, the existence of a deterministic polynomial-time algorithm $\A$ for 
Fixed-Size $\CliqueCover{4}$ would imply the existence of a
deterministic polynomial-time algorithm for $\CliqueCover{4}$. The input
for the hypothetical $\A$ consists of a graph $G$ together with four positive
integers $\Set{n_1,n_2,n_3,n_4}$. For the desired 
reduction, we use the function $p(n,4)$, a cubic polynomial
in $n$, which counts the number of ways of partitioning positive 
integer $n$ into $4$ positive integers~\cite{andrews1998}:
\[
      p(n,4)\ :=
      \ \begin{cases}
        \Bigl[ \dfrac{(n+1)^3}{144} - \dfrac{(n+1)}{48} \Bigr]
         \qquad & \text{if $n$ is even},
        \\[1.9ex]
        \Bigl[ \dfrac{(n+1)^3}{144} - \dfrac{(n+1)}{12} \Bigr]
         \qquad & \text{if $n$ is odd},
        \end{cases}
\]
where $[x]$ is the nearest integer to $x$. We leave to the reader the 
straightforward task of writing an algorithm $\A'$ to generate all partitions 
of $n$ into $4$ positive integers, which runs in $\bigOO{n^3}$ time.
To decide whether an arbitrarily given graph $G$ is a positive instance of 
$\CliqueCover{4}$, we run algorithm $\A'$ to generate the successive $4$-part 
partitions of $n = \size{V(G)}$. The given $G$ has a $4$-clique
cover iff algorithm $\A$ returns ``yes'' when its input is: $G$
together with at least one of these $4$-part partitions of integer $n$.
\end{proof}

\begin{lemma}
\label{lem:equal-size-4-clique-cover}
Equal-Size $\CliqueCover{4}$ is NP-complete.
\end{lemma}

\begin{proof}
If $\Set{A_1,A_2,A_3,A_4}$ is a $4$-part partition of $V(G)$, where
$\size{A_1}=\size{A_2}=\size{A_3}=\size{A_4} = \size{V(G)}/4 = n/4 $, 
we can verify in polynomial time that the induced graphs
$\Set{G[A_1],G[A_2],G[A_3],G[A_4]}$ are each complete. 
So the problem is in NP.

NP-completeness follows by reduction from Fixed-Size $\CliqueCover{4}$
to Equal-Size $\CliqueCover{4}$, \ie, the existence of a deterministic
polynomial-time algorithm $\A$ for Equal-Size $\CliqueCover{4}$ would 
imply the existence of a deterministic 
polynomial-time algorithm for Fixed-Size $\CliqueCover{4}$, as shown
next. Given an arbitrarily given graph $G$ and four positive integers
$\Set{n_1,n_2,n_3,n_4}$ such that $n_1+n_2+n_3+n_4 = \size{V(G)} = n$,
we introduce $4$ new sets of vertices $\Set{A_1,A_2,A_3,A_4}$ such that:
\[
  \size{A_1} = n - n_1, \quad \size{A_2} = n - n_2,
  \quad \size{A_3} = n - n_3,  \quad \size{A_4} = n - n_4.
\]
We construct a new graph $G'$ such that:
\begin{alignat*}{4}
    &  V(G')\ && :=\ && V(G) \cup A_1 \cup A_2 \cup A_3 \cup A_4 ,
\\[1.5ex]     
    &  E(G')\ && :=\ && 
      E(G) \cup \Set{\,\set{v\,w}\;|
           \; v\in V(G)\text{ and } w\in A_1\cup A_2\cup A_3\cup A_4\,} 
\\
   & && && \cup \Set{\,\set{v\,w}\;|\;v,w\in A_1\,}
      \cup \Set{\,\set{v\,w}\;|\;v,w\in A_2\,}
      \cup \Set{\,\set{v\,w}\;|\;v,w\in A_3\,}
      \cup \Set{\,\set{v\,w}\;|\;v,w\in A_4\,} .
\end{alignat*}
In words, the new $G'$ is obtained from $G$ by adding four cliques, one
clique on each of the new vertex sets in $\Set{A_1,A_2,A_3,A_4}$, and by
connecting every vertex in $A_1\cup A_2\cup A_3\cup A_4$ with every
vertex in $V(G)$.  Hence, $G'$ has $4n$ vertices, and there
are no edges between the subgraphs
$\Set{G[A_1],G[A_2],G[A_3],G[A_4]}$.

Using the fact that no clique in $G'$ can contain vertices from two
distinct sets in $\Set{A_1,A_2,A_3,A_4}$, we conclude that $G$ is a
positive instance of Fixed-Size $\CliqueCover{4}$ with sizes
$\Set{n_1,n_2,n_3,n_4}$ iff $G'$ is a positive instance of Equal-Size
$\CliqueCover{4}$.
\end{proof}

For the rest of this section we restrict attention to graphs $G$ in the
class $\graphs{2^{\bm{*}}}$, introduced before Definition~\ref{def:min-bisection}.
A balanced reassembling $(G,\B)$ of such a graph $G$ is relative to
a full binary tree $\B$ of height $\log n$ (\ie,
with $1+\log n$ levels) where $n = \size{V(G)}$.

The nodes in a reassembling tree $\B$ are each a cluster of vertices
(Definition~\ref{defn:binaryTrees}), a subset of $V(G)$. One
of the measures on a node/cluster $X$ in the reassembling $(G,\B)$ is
its height (Definition~\ref{defn:graph-reassembling}), denoted
$\height{G,\B}{X}$. 
We can extend the measure $\heightSym$ to every
edge $\set{v\, w}\in E(G)$, by defining: 
\[
   \height{G,\B}{\set{v\, w}}\ :=\ \min\,\SET{\,\height{G,\B}{X}\;\bigl|
     \; X\in\B \text{ and both } v, w\in X\,} .
\]
In words, the height of $\set{v\, w}$ in $(G,\B)$ is the
height of the least node/cluster $X$ that includes both endpoints of
edge $\set{v\, w}$. Note that, if $\B$ is a full binary tree (the 
case of a balanced reassembling of $G\in\graphs{2^{\bm{*}}}$),
then the node/cluster $X$ is at the same distance (or height)
from the endpoints (or leaf nodes) $v$ and $w$.

Another way of understanding $\height{G,\B}{\set{v\, w}} = p$, where 
$0\leqslant p\leqslant \log n$, is that $p$ is the
level number in $\B$ (starting from the bottom, with level $0$ being the
level of all leaf nodes) at which the two halves of edge $\set{v\, w}$ are
spliced together, or at which the edge $\set{v\, w}$ is re-introduced
in the reassembling.

\begin{lemma}
\label{lem:different-formula-for-beta-measure}
If $G\in\graphs{2^{\bm{*}}}$ and $(G,\B)$ is a balanced reassembling of $G$,
then:
\[
   \beta(G,\B)\ =
  \ 2\times \sum\,\Set{\,\height{G,\B}{\set{v\, w}}\;|\;\set{v\, w}\in E(G)\,}.
\]
Informally, $\beta(G,\B)$ is minimized (resp. maximized) when
edges are spliced as soon as possible (as late as possible) in the
reassembling.
\end{lemma}

\begin{proof}
Straightforward consequence of Definitions~\ref{defn:graph-reassembling}
and~\ref{def:different-measures}. A formal proof can be carried
out by induction on $p\geqslant 1$, where $\size{V(G)} = n = 2^p$.
All details omitted.
\end{proof}

The proof of the next lemma is interesting in that it combines both
algebraic reasoning (formulation of an instance of
integer quadratic programming) and combinatorial reasoning (existence 
of a partition of $V(G)$ into four independent vertex sets of equal size).

\begin{lemma}
\label{lem:four-disjoint-independent-sets}
Let $G\in\graphs{2^{\bm{*}}}$ with $\size{V(G)} = n = 2^p$ for some
$p\geqslant 2$. Let $A_1,A_2,A_3,A_4\subseteq V(G)$ be four disjoint
independent sets of vertices in $G$, each of size $n/4$.
Let $(G,\B)$ be a balanced reassembling and 
$X_1,X_2,X_3,X_4\in\B$ be the 
nodes/vertex-clusters in the tree $\B$ such that 
$\size{X_1} =\size{X_2} =\size{X_3} =\size{X_4} = n/4$.%
   \footnote{
 Stated differently, $\Set{X_1,X_2,X_3,X_4}$ are the four nodes of $\B$ 
 such that:
 \[
  \height{G,\B}{X_1}=\height{G,\B}{X_2}=\height{G,\B}{X_3}
  =\height{G,\B}{X_4}= p-2,
 \]
 where $p = \log n$ with $n = \size{V(G)}$, 
 \ie, $\Set{X_1,X_2,X_3,X_4}$ are the four grandchildren of the
 root node $V(G)$.
   }

\medskip\noindent
\textbf{Conclusion:} 
If the $\beta$-measure $\beta(G,\B)$ is maximized, \ie,
\[
  \beta(G,\B)\ =\ \max\,\Set{\,\beta(G,{\B}')\;|
  \; {\B}'\text{ is a balanced reassembling tree }\,},
\]
then $\Set{X_1,X_2,X_3,X_4}$ are disjoint independent sets in $G$,
not necessarily the same as $\Set{A_1,A_2,A_3,A_4}$, 
each of size $n/4$.
\end{lemma}

\begin{proof}
Let $V(G) = \Set{v_1,\ldots,v_n}$. Assume $(X_1\uplus X_2)$ and $(X_3\uplus X_4)$
are the two children-nodes/vertex-clusters of the root $V(G)$ in $\B$. The
height in $(G,\B)$ of every vertex cluster in $\Set{X_1,X_2,X_3,X_4}$ is $(p-2)$,
while the height of both $(X_1\uplus X_2)$ and $(X_3\uplus X_4)$ is $(p-1)$,
and the height of the root $V(G)$ is of course $p$.

Lemma~\ref{lem:different-formula-for-beta-measure} gives an alternative
definition of $\beta(G,\B)$, obtained by summing the heights in $(G,\B)$
of all the edges. This definition is presumed in the formulation 
of the integer quadratic programming below.  The problem of finding a 
balanced reassembling tree $\B$ which 
maximizes $\beta(G,\B)$ can be translated to an integer
quadratic programming, as follows:
\begin{equation*}
\begin{array}{ll@{}ll}
 \text{maximize} 
    & \text{(i)}\quad 
    & \displaystyle{2p \times \sum_{\set{v_i\,v_{j}}\ \in\ E(G)} 
    (x_{i,1}x_{j,3}\ +\ x_{i,1}x_{j,4}\ +\ x_{i,2}x_{j,3}\ +\ x_{i,2}x_{j,4})}
    & + 
\\[2.5ex]
    & \text{(ii)}\quad 
    & \displaystyle{2(p-1) \times \sum_{\set{v_i\,v_{j}}\ \in\ E(G)} 
    (x_{i,1}x_{j,2}\ +\ x_{i,3}x_{j,4} )} & + 
\\[2.5ex]
    & \text{(iii)}\quad 
    & \displaystyle{ 2(p-2) \times \sum_{\set{v_i\,v_{j}}\ \in\ E(G)} 
    (x_{i,1}x_{j,1}\ +\ x_{i,2}x_{j,2}\ +\ x_{i,3}x_{j,3}\ +\ x_{i,4}x_{j,4}) }
    &
% \\
\end{array}
\end{equation*}
\begin{equation*}
\begin{array}{ll@{}ll}
 \text{subject to} 
  &\text{(i)}\quad 
  & \text{for all $1\leqslant i\leqslant n$ and
              $1 \leqslant k \leqslant 4$,} \quad
  & \quad x_{i,k} \in \Set{0,1}
\\[2.5ex]
  &\text{(ii)}\quad 
  &\text{for all $1 \leqslant k \leqslant 4$,} \quad
  &\quad  \displaystyle{\sum_{1 \leqslant i \leqslant n} x_{i,k} = \dfrac{n}{4}}
\\[2.5ex]
  &\text{(iii)}\quad 
  & \text{for all $1 \leqslant i \leqslant n$,} \quad
  &\quad  \displaystyle{\sum_{1 \leqslant k \leqslant 4} x_{i,k} = 1}
\end{array}
\end{equation*}
The optimization objective is quadratic and has three parts with
respective coefficients $2p$, $2(p-1)$, and $2(p-2)$, 
while the constraints are linear. Every vertex $v_i \in V(G)$ 
corresponds to four variables 
$\Set{x_{i,1},x_{i,2},x_{i,3},x_{i,4}}$, which indicate which set in 
$\Set{X_1,X_2,X_3,X_4}$ contains vertex $v_i$. Specifically, 
for every $1\leqslant i\leqslant n$ and $1\leqslant k \leqslant 4$:
\[
 x_{i,k}\ = \ \begin{cases}
           1   \qquad    & \text{if $v_i \in X_k$,} 
           \\
           0       & \text{if $v_i \notin X_k$.}
           \end{cases}
\]
To understand the preceding formulation as a $(0,1)$ quadratic
programming, observe that for every edge $\set{v_i\,v_{j}}$:
\[
    x_{i,k}\;x_{j,\ell}\ =
    \ \begin{cases}
      1 \qquad & \text{if $v_i\in X_k$ and $v_{j}\in X_{\ell}$,}
      \\
      0 \qquad & \text{if $v_i\not\in X_k$ or $v_{j}\not\in X_{\ell}$.}
      \end{cases}
\]
Note that the suggested quadratic system is slightly relaxed in the sense that
(as represented by part (iii) with coefficient $2 (p-2)$ 
of the optimization objective) we assume that 
every edge $\set{v_i\,v_{j}}$ whose endpoints are in the same $X_k$,
\ie, $v_i, v_{j} \in X_k$ for some $1 \leqslant k \leqslant 4$, 
contributes the maximum possible value, here $2(p-2)$, rather than
its exact value to $\beta(G,\B)$. The rest of the proof shows that 
this relaxation does not affect its correctness.

A straightforward re-ordering of terms shows that the optimization
objective can be written as follows:
\begin{alignat*}{5}
   & \theta\ &&:=\ && 
     2p \times \Bigl(\ \sum_{\set{v_i\,v_j}\;\in\;E(G),\ 1 \leqslant k\leqslant\ell\leqslant 4}
         {x_{i,k}\;x_{j,\ell}}\ \Bigr) - 2 \theta_1 - 4 \theta_2
     \qquad &&\text{where}
\\[2ex]
   & \theta_1\ &&:=\ && 
     \sum_{\set{v_i\,v_{j}}\;\in\; E(G)} (x_{i,1}\;x_{j,2}\ +\ x_{i,3}\;x_{j,4} )
     \qquad &&\text{and}
\\[2ex]
   & \theta_2\ &&:=\ && 
     \sum_{\set{v_i\,v_{j}}\;\in\; E(G),\ 1 \leqslant k\leqslant 4} x_{i,k}\;x_{j,k}  \ .
\end{alignat*}
The quantity $\theta_1$ counts the number of edges $\set{v_i\,v_{j}}$
satisfying one of two conditions:
\begin{itemize}[itemsep=0pt,parsep=3pt,topsep=3pt,partopsep=0pt]
\item either the endpoints $v_i$ and $v_j$ are in $X_1$ and $X_2$,
\item or the endpoints $v_i$ and $v_j$ are in $X_3$ and $X_4$.
\end{itemize}
Every edge satisfying one of the two preceding conditions has height
$(p-1)$ in $(G,\B)$.  The quantity $\theta_2$ counts the number of
edges $\set{v_i\,v_{j}}$ whose endpoints $v_i$ and $v_j$ are in the
same $X_k$, and whose height is therefore $\leqslant (p-2)$ in $(G,\B)$.
We now observe that:
\[
  \sum_{\set{v_i\,v_j}\;\in\;E(G),\ 1 \leqslant k\leqslant\ell\leqslant 4}
         {x_{i,k}\;x_{j,\ell}}\ =\ \size{E(G)}\ =\ m .
\]
The optimization objective can now be simplified to read:
\begin{alignat*}{5}
    & \theta\ && =\ && 2pm\ -\ 2\theta_1\ -\ 4\theta_2
      \ =\ 2(p-1)m\ +\ 2(m-\theta_1-\theta_2)\ -\ 2\theta_2 
      \ =\ 2(p-1)m\ + \theta' \quad && \text{where}
\\[1ex]
   & \theta'\ && :=\ && 2(m-\theta_1-\theta_2)\ -\ 2\theta_2 \ .
\end{alignat*}
Hence, maximizing $\theta$ is equivalent to maximizing $\theta'$, and the
latter is maximized when $(m-\theta_1-\theta_2)$ is maximized and $\theta_2$
is minimized. But $(m-\theta_1-\theta_2)$ is the number of edges whose
height in $(G,\B)$ is $p$, \ie, the edges $\set{v_i\,v_j}$ such that
$v_i\in (X_1\uplus X_2)$ and $v_j\in (X_3\uplus X_4)$, while 
$\theta_2$ is the number of edges $\set{v_i\,v_{j}}$ whose endpoints 
$v_i$ and $v_j$ are in the same $X_k$ and whose height is $\leqslant (p-2)$.

We switch to combinatorial reasoning, by invoking the fact that $G$
has four disjoint independent sets, each with $(n/4)$ vertices. There is
no need to explicitly solve the integer quadratic programming above.
Maximizing $\theta'$ means choosing $(X_1 \uplus X_2, X_3 \uplus X_4)$
as a bisection with a maximum cut 
$\size{\bridges{}{X_1 \uplus X_2, X_3 \uplus X_4}}$, \ie, choosing 
$(X_1 \uplus X_2)$ and $(X_3 \uplus X_4)$ as independent sets. And
minimizing $\theta_2$, which is here possible down to $\theta_2 = 0$,
means choosing each of $X_1$, $X_2$, $X_3$, and $X_4$, as an independent set.
\end{proof}

\begin{lemma}
\label{lem:reduction-from-equal-size-4-clique-cover}
Let $G\in\graphs{2^{\bm{*}}}$ be
a positive instance of Equal-Size $\CliqueCover{4}$, 
with $\size{V(G)} = n = 2^p$ for some $p\geqslant 2$.
Let $(G,\B)$ be a balanced reassembling and 
$\Set{X_1,X_2,X_3,X_4}\subseteq\B$ be the 
nodes/vertex-clusters in the tree $\B$ such that 
$\size{X_1} =\size{X_2} =\size{X_3} =\size{X_4} = n/4$.

\medskip\noindent
\textbf{Conclusion:} 
If $(G,\B)$ is a $\beta$-optimal balanced reassembling, then
$\Set{X_1,X_2,X_3,X_4}$ is an \\ Equal-Size $\CliqueCover{4}$ of $G$.
\end{lemma}

\begin{proof} Because $(G,\B)$ is a $\beta$-optimal balanced reassembling
(Definition~\ref{def:different-measures}),
we have:
\[
  \beta (G,\B)\ := \ \min\,
  \SET{\,\beta (G,{\B}')\;\bigl|
  \;{\B}'\text{ is a balanced binary tree over $V(G)$}\,}.
\]
We write $K_n$ for the complete graph over $n = \size{V(G)}$ vertices,
and $\overline{G}$ for the complement of $G$. We thus have
$V(\overline{G}) = V(G)$ and $E(\overline{G}) = E(K_n) - E(G)$.
Clearly, $G = \overline{\overline{G}}$ and we can write 
$G = K_n - \overline{G}$. Hence:
\[
  \beta (G,\B)\ = \ \min\,
  \SET{\,\beta (K_n - \overline{G},{\B}')\;\bigl|
  \;{\B}'\text{ is a balanced binary tree over $V(G)$}\,}.
\]
The $\beta$-measure of a balanced reassembling of $K_n$, call it $M$, does
\emph{not} depend on the reassembling tree, \ie, for all balanced reassembling
trees ${\B}_1$ and ${\B}_2$ on $n$ vertices, it holds that:%
    \footnote{We do not need the exact value of $M$ for this proof,
    it suffices to know it exists, which is an obvious consequence
    of the fact that every bijection from ${\B}_1$ to ${\B}_2$ produces
    an isomorphism (Definition~\ref{defn:isomorphic-graph-reassemblings})
    between the balanced reassemblings $(K_n,{\B}_1)$ and $(K_n,{\B}_2)$.    
    It takes some effort to compute $M$ precisely (omitted here): 
    For all balanced reassembling trees $\B$, if $n = 2^p$,  
    it can be shown that 
    $M = \beta (K_n,\B) = (p-1)\times 2^{2p} + 2^p = ((\log n)-1)\times n^2 + n$.
    }
\[
   \beta (K_n,{\B}_1)\ =\ \beta (K_n,{\B}_2)\ =\ M.
\]
Hence, the following equality holds:
\[
  \beta (G,\B)\ = \ M - \max\,
  \SET{\,\beta (\overline{G},{\B}')\;\bigl|
  \;{\B}'\text{ is a balanced binary tree over $V(G)$}\,}.
\]
In words, the value of $\beta (G,\B)$ is \emph{minimized} when the value
of $\beta (\overline{G},{\B}')$ is \emph{maximized}. Since
$(G,\B)$ is $\beta$-optimal, $\beta (\overline{G},{\B}')$ is \emph{maximized}.

By hypothesis, $G$ is an instance of Equal-Size $\CliqueCover{4}$, 
\ie, there are disjoint sets $\Set{A_1, A_2, A_3, A_4}$ of vertices in $G$ such 
that the induced subgraphs in $\Set{G[A_1],G[A_2],G[A_3],G[A_4]}$ are each
a complete graph with $n/4$ vertices. Hence, the 
corresponding subgraphs in $\overline{G}$, namely 
$\Set{\overline{G}[A_1],\overline{G}[A_2],\overline{G}[A_3],\overline{G}[A_4]}$,
are each an edgeless graph with $n/4$ vertices. Equivalently,
$\Set{A_1, A_2, A_3, A_4}$ are disjoint independent sets 
in $\overline{G}$, each with $n/4$ vertices. Hence, by 
Lemma~\ref{lem:four-disjoint-independent-sets}, $\Set{X_1,X_2,X_3,X_4}$ are 
disjoint independent sets in $\overline{G}$, and therefore disjoint cliques 
in  $G$,each with $n/4$ vertices.
\Hide{
To conclude the proof, it suffices to observe that for a balanced
reassembling $(\overline{G},{\B}')$, the measure
$\beta (\overline{G},{\B}')$ is \emph{maximized} when the four 
grandchildren nodes/clusters, say $\Set{X'_1,X'_2,X'_3,X'_4}$,
of the root $V(G)$ in ${\B}'$ induce subgraphs 
$\Set{\overline{G}[X'_1],\overline{G}[X'_2],
\overline{G}[X'_3],\overline{G}[X'_4]}$ which are each an edgeless graph
with $n/4$ vertices. Indeed, in order to maximize $\beta (\overline{G},{\B}')$,
we have to maximize $\height{\overline{G},{\B}'}{\set{v\, w}}$ for
all $\set{v\,w}\in E(\overline{G})$, 
by Lemma~\ref{lem:different-formula-for-beta-measure}. Hence,
every edge $\set{v\,w}$ must connect $v\in V(\overline{G}[X'_i])$
and $w\in V(\overline{G}[X'_j])$ for some $i\neq j$ and be such that
$\height{\overline{G},{\B}'}{\set{v\, w}} = p$ or $p-1$ where $2^p = n
= \size{V(G)}$. }
\end{proof}

\begin{theorem}
\label{thm:beta-optimization}
For the class of simple undirected graphs $G$, the computation of 
$\beta$-optimal balanced reassemblings $(G,\B)$ is an NP-hard problem.
\end{theorem}

\begin{proof}
If a deterministic polynomial-time algorithm existed for producing a 
$\beta$-optimal balanced reassembling $(G,\B)$ of an arbitrarily
given $G$, then this algorithm could be
used again to decide in deterministic polynomial-time whether a graph in 
$\graphs{2^{\bm{*}}}$ is a positive instance of Equal-Size $\CliqueCover{4}$,
by Lemma~\ref{lem:reduction-from-equal-size-4-clique-cover}.
This would in turn contradict Lemma~\ref{lem:equal-size-4-clique-cover}
asserting the NP-completeness of Equal-Size $\CliqueCover{4}$. The
desired conclusion follows.
\end{proof}

\section{Related and Future Work}
\label{sect:future}
    %% future.tex

As mentioned in Section~\ref{sect:intro}, graph reassembling
can be considered as a special case of a family of graph
embedding problems, known as \emph{communication tree embedding} problems.
Communication tree embedding is the problem of
embedding the vertices of a source graph $G$ into the
nodes of a host tree $T$.  In the case where the vertices of $G$ are
mapped into the leaves of the host tree, the underlying tree is called
a \emph{routing tree} (or call routing tree) and the related problems are referred to as
routing tree embedding problems.  Graph reassembling is a
slight variation of a special case of routing tree embedding
where the internal nodes of the host tree have each degree 3, known
as \emph{tree layout problem}.

In the context of communication tree embedding, and more
specifically tree layout problem, different measures are
defined.  Corresponding to our $\alpha$ measure in this report
is the \emph{edge congestion} measure, which represents the maximum
communication traffic on the edges of the host tree $T$.
Seymour and Thomas in ~\cite{seymour1994call} show that minimum congestion
routing tree problem, referred to as the \emph{minimum carving-width}
problem, is solvable in $O(n^4)$ for planar graphs, but is
NP-hard in general.
The efficiency of the method of Seymour and Thomas for planar graphs, was
improved in~\cite{hicks2005planar} and later in~\cite{gu2008optimal}.
While the result of~\cite{seymour1994call} for
minimum carving width of planar graphs can be extended to
$\alpha$-optimal \emph{general} graph reassembling, the status of
$\alpha$-optimal \emph{balanced} graph reassembling for planar graphs
is open, which we also conjecture to be NP-hard in contrast to the case
of general graph reassembling for planar graphs.

Our $\beta$ measure in this report corresponds to another measure,
known as \emph{tree length}, which is equivalent to the summation of edge
congestions in the underlying host tree $T$.  The tree length of a
tree embedding represents the average delay (\ie, average
dilation) in the source graph and equivalently the average traffic on
the edges of the host tree.  In~\cite{mirzaei2016AvgDelay} it is shown
that finding a tree layout with minimum tree length is NP-hard when
the source graph is a general graph with no self loops. In the
same report also the NP-hardness result is extended to the more
general routing tree embedding problem.  There are several other
measures defined for different variations of communication tree
problems; a comprehensive list of these measures can be found
in~\cite{i2001layout, alvarez2007communication, petit:2011}.

Future work related to graph reassembling (and specifically related to
the balanced case) includes a study of classes of graphs for which
$\alpha$-optimization and/or $\beta$-optimization of their
reassembling can be carried out in low-degree polynomial times.  On the other hand,
as indicated in our earlier report~\cite{kfouryReassembling}, the smaller
the $\alpha$ and $\beta$ measures are, the more efficient the
execution of programs is, in a domain-specific language (DSL) for the
design of flow networks~\cite{BestKfoury:dsl11,Kfoury:SCP2014,
SouleBestKfouryLapets:eoolt11,kfoury2013different}. Beinstock in~\cite{bienstock1990embedding}
presents some elementary classes of graphs with small optimal tree congestion as well as an
upper bound for the value of optimal tree congestion based on the
\emph{tree decomposition} of the graphs.
Deciding weather the width of tree decomposition of an arbitrary graph
is at most $k$ is NP-complete~\cite{arnborg1987complexity}, however the problem is tractable
for small and fixed values of $k$. Bodlaender in~\cite{bodlaender1994tourist} surveys a list of graph classes for which the
treewidth can be computed in polynomial time.
Similarly relating tree width to minimum tree congestion of graphs,
\cite{belmonte2012characterizing} characterizes the
class of graphs that have $\alpha$-measure at most $3$,
but extending these result to different variety of graph reassembling such as balanced case and
also study of the characterization of the classes of graphs with higher value of $\alpha$-measure
is open to be investigated in future.

Finally, there is the question of computing approximations of
$\alpha$-optimal and $\beta$-optimal balanced graph reassembling, whereby we
can turn the NP-hardness of any of the preceding optimizations into
polynomially-solvable optimizations. In~\cite{khuller1994designing}
the problem of finding call routing trees with minimum congestion
has been studied and an approximate method for finding a solution within
a $O(\log n)$ factor from the optimal solution is suggested.
Hence the natural question is if similar methods can be facilitated in order to find
approximation of $\alpha$-optimal and $\beta$-optimal for balanced reassembling as
well as other variations of graph reassembling.

\Hide
{\footnotesize
\printbibliography
}

% \Hide
{\footnotesize % \small
\bibliographystyle{plainurl} % {plainnat} % {plain} % {alpha} % {siam} % {abbrv}
   % plain, abbrv, siam, alpha, are among dozens of other available styles
\bibliography{generic,extra}
}

\ifTR
%  \clearpage
\else
\fi

\end{document}